\documentclass[runningheads]{llncs}

\usepackage{amsmath}
\usepackage{amssymb}
\usepackage{enumerate}

\usepackage{booktabs}
\usepackage{array}
\usepackage{makecell}
\usepackage{stackengine}
\usepackage{relsize}

\usepackage{ifthen}
\usepackage{stmaryrd}

\usepackage{lmodern}

\usepackage{tikz}
\usetikzlibrary{arrows,automata,positioning}

\usepackage[hypertexnames=false]{hyperref}
\hypersetup{
    colorlinks,
    linkcolor={red!40!black},
    citecolor={blue!50!black},
    urlcolor={blue!80!black}
}

\newcommand{\visiblehref}[2]{\href{#1}{#2}\footnote{\url{#1}}}

\usepackage{algorithm}
\usepackage[italicComments=false,indLines=false]{algpseudocodex}
\algrenewcommand{\alglinenumber}[1]{\color{gray}\scriptsize#1:}

\usepackage{cleveref}

\usepackage[firstpage]{draftwatermark}

\SetWatermarkAngle{0}
\SetWatermarkText{\raisebox{14.5cm}{%
  \hspace{0.1cm}%
  \href{https://doi.org/10.5281/zenodo.7885130}{\includegraphics{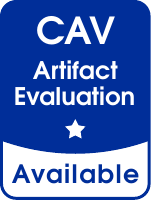}}%
  \hspace{9cm}%
  \includegraphics{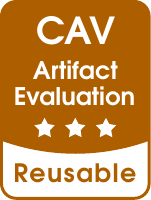}%
}}

\SetSymbolFont{stmry}{bold}{U}{stmry}{m}{n}

\apptocmd{\sloppy}{\hbadness 10000\relax}{}{}

\makeatletter
 \def\@textbottom{\vskip \z@ \@plus 2pt}
 \let\@texttop\relax
\makeatother

\newcommand{\tinyskip}{\vspace{2.5pt plus 0.5pt minus 0.5pt}}
\renewcommand{\paragraph}[1]{\tinyskip \noindent \textbf{#1}~~}

\newcommand{\code}[1]{{\normalfont \texttt{#1}}}

\spnewtheorem{definition}{Definition}{\bfseries}{}

\newcommand{\Continue}{\textbf{continue}}
\newcommand{\Returning}[1]{\textbf{returning} #1}

\newcommand{\naive}{naïve}
\newcommand{\Naive}{Naïve}

\tikzset{
    ->, %
    every edge/.style={draw, thick},
    >=stealth,
    node distance=1.5cm, %
    every state/.style={minimum size=0pt, inner sep=3pt, thick}, %
    open/.style={dashed},
}

\newcommand{\A}{\mathcal{A}}
\newcommand{\tr}[3]{{#1}{\xrightarrow{#2}}{#3}}
\newcommand{\TT}{\centerdot}
\newcommand{\derName}{\delta}
\newcommand{\der}[2][{}]{\derName_{#1}(#2)}

\newcommand{\ite}[3]{\boldsymbol{(}#1\,{\boldsymbol{?}}\,#2\,{\boldsymbol{:}}\,#3\boldsymbol{)}}
\newcommand{\st}{^{\texttt{*}}}
\newcommand{\IfThenElse}[3]{\mathbf{if}\, #1\, \mathbf{then}\, #2\, \mathbf{else}\, #3}
\newcommand{\Nullable}[1]{\textit{Nullable}(#1)}
\newcommand{\Tail}[2][1]{#2_{#1..}}
\newcommand{\ith}[2]{#1_{#2}}
\newcommand{\OutDegree}[1]{\textit{deg}^+\!(#1)}

\newcommand{\eps}{\varepsilon}
\newcommand{\emp}{\bot}
\newcommand{\rand}{\,\texttt{\&}\,}
\newcommand{\ror}{\,\texttt{|}\,}
\newcommand{\rnot}{\texttt{\textasciitilde{}}}

\newcommand{\REPRED}{\textit{RE}}
\newcommand{\rden}[1]{{\ifthenelse{\equal{#1}{}}{{\mathbf{L}}}{\mathbf{L}(#1)}}}

\newcommand{\deriv}[2][{}]{{\ifthenelse{\equal{#2}{}}{{\partial_{#1}}}{\partial_{#1}(#2)}}}

\newcommand{\normop}[1]{\stackMath\mathbin{\stackinset{c}{0ex}{c}{0ex}{#1}{\textrm{\small$\bigcirc$}}}}
\newcommand{\Nrand}{\normop{\rand}}
\newcommand{\Nror}{\normop{\ror}}
\newcommand{\Ncdot}{\normop{\cdot}}
\newcommand{\Nrnot}{\normop{\rnot}}

\begin{document}

\title{Incremental Dead State Detection in Logarithmic Time (Extended Version)}

\author{Caleb Stanford\inst{1}
\and Margus Veanes\inst{2}}
\institute{
University of California, Davis
\email{cdstanford@ucdavis.edu}
\and
Microsoft Research, Redmond
\email{margus@microsoft.com}
}

\maketitle

\begin{abstract}
Identifying live and dead states in an abstract transition system
is a recurring problem in formal verification;
for example, it arises in our recent work on efficiently deciding regex constraints in SMT.
However, state-of-the-art graph algorithms for maintaining reachability information \emph{incrementally}
(that is, as states are visited and before the entire state space is explored)
assume that new edges can be added from any
state at any time, whereas in many applications,
outgoing edges are added from each state as it is explored.
To formalize the latter situation, we propose \emph{guided incremental digraphs} (GIDs),
incremental graphs which support labeling \emph{closed} states
(states which will not receive further outgoing edges).
Our main result is that dead state detection in GIDs
is solvable in $O(\log m)$ amortized time per edge for $m$ edges,
improving upon $O(\sqrt{m})$ per edge due to Bender, Fineman, Gilbert, and Tarjan (BFGT) for general incremental directed graphs.

\medskip

We introduce two algorithms for GIDs: one establishing the logarithmic time bound, and a second algorithm to explore a lazy heuristics-based approach.
To enable an apples-to-apples experimental comparison,
we implemented both algorithms, two simpler baselines, and the state-of-the-art BFGT baseline using a common directed graph interface in Rust.
Our evaluation shows $110$-$530$x speedups over BFGT for the largest input graphs over a range of graph classes, random graphs, and graphs arising from regex benchmarks.

\end{abstract}

\section{Introduction}
\label{sec:intro}

Classifying states in a transition system as live or dead is a recurring problem in formal verification. For example,
given an expression, can it be simplified to the identity?
Given an input to a nondeterministic program, can it reach a terminal state, or can it reach an infinitely looping state?
Given a state in an automaton, can it reach an accepting state?
State classification is relevant to satisfiability modulo theories (SMT) solvers~\cite{BM08,de2011satisfiability,CVC4Paper,CVC5},
where theory-specific partial decision procedures often work by exploring the state space to find a reachable path that corresponds to a satisfying string or, more generally, a sequence of constructors.
To a first approximation, the core problem in all of these cases
amounts to classifying each state $u$ in a directed graph as \emph{live},
meaning that a feasible, accepting, or satisfiable state is reachable from $u$;
or \emph{dead}, meaning that all states reachable from $u$ are infeasible, rejecting, or unsatisfiable.

\paragraph{Motivating applications.}
We originally encountered the problem of incremental state classification during our prior work while building Z3's regex solver~\cite{stanfordsymbolic}
for the SMT theory of string and regex constraints~\cite{SMT12-regex,berzish2021smt,amadini2021survey}.
Our solver leveraged \emph{derivatives} (in the sense of Brzozowski~\cite{Brzozowski64} and Antimirov~\cite{Ant95}) to explore the states of the finite state machine corresponding to the regex incrementally (as the graph is built), to avoid the prohibitive cost of expanding all states initially.
This turns out to require solving the live and dead state detection problem in the finite state machine presented as an incremental directed graph.\footnote{The specific setting is regexes with intersection and complement (\emph{extended}~\cite{kupferman2002improved,gelade2008succinctness} or \emph{generalized}~\cite{ellul2005regular} regexes), which are found natively in security applications~\cite{Zelkova,stanfordsymbolic}.
Other solvers have also leveraged derivatives~\cite{LTRTB15}
and laziness in general~\cite{hooimeijer2010solving}.
}
Concretely, consider the regex $(\TT\st \alpha \TT^{100})^C \cap (\TT \alpha)$,
where $\TT$ matches any character, $\cap$ is regex intersection,
$^C$ is regex complement,
and $\alpha$ matches any digit ($0$-$9$).
A traditional solver would expand the left and right operands as state machines,
but the left operand $(\TT\st \alpha \TT^{100})^C$ is astronomically large as a DFA, causing the solver to hang.
The derivative-based technique instead constructs the derivative regex:
$(\TT\st \alpha \TT^{100})^C \cap (\TT^{100})^C \cap \alpha$.
At this stage we have a graph of two states and one edge, where the states are the two regexes just described, and the edge is the derivative relation.
After one more derivative operation, the regex is reduced to one that is clearly nonempty
as it accepts the empty string.

It is important that a derivative-based solver identify nonempty (live) and empty (dead) regexes \emph{incrementally}
because it does not generally construct the entire state space before terminating
(see the graph update rule \textsc{Upd}, p. 626~\cite{stanfordsymbolic}).
Moreover, the nonemptiness problem for extended regexes is non-elementary~\cite{stockmeyer1973word} --- and still PSPACE-complete for more restricted fragments ---
which strongly favors a lazy approach over brute-force search.

Regexes are just one possible application; the algorithms we will present here are broadly applicable to any context where the states have a bounded (per-node) out-degree.
For example, they could be applied in LTL model checking when lazily exploring the state space of a nondeterministic Büchi automaton (NBA), where the NBA is too expensive to construct up front.
The important fact is that each state of the automaton has only finitely many outgoing edges, and when all these are added, we can hope to check for dead states incrementally.

\paragraph{Prior work.}
Traditionally, while live state detection can be done incrementally, dead state detection is often done exhaustively (i.e., after the entire state space is explored).
For example, bounded and finite-state model
checkers
based on translations to
automata~\cite{clarke1994another,kupferman2001model,rozier2007ltl},
as well as classical dead-state elimination algorithms~\cite{HopUll79,Blum96,BBCF11},
typically work on a fixed state space after it has been fully enumerated.
However, we reiterate that
exhaustive exploration is prohibitive for large (e.g., exponential or infinite) state spaces which arise in an SMT verification context.
We also have good evidence that incremental feedback can improve SMT solver performance:
a representative success story is
the e-graph data structure~\cite{de2007efficient,willsey2021egg},
which maintains an equivalence relation among expressions incrementally;
because it applies to general expressions, it is theory-independent and re-usable.
Incremental state space exploration could lead to similar benefits if applied
to SMT procedures which still rely on exhaustive search.

However, in order to perform incremental dead state detection,
we currently lack algorithms which match offline performance.
As we discuss in \Cref{sec:graph}, the best-known existing solutions would require maintaining strong connected
components (SCCs) incrementally.
For SCC maintenance and the related simpler problem of cycle detection,
amortized algorithms are known with
$O(m^{3/2})$ total time for $m$ edge additions~\cite{HKMST12,BFGT15},
with some recently announced
improvements~\cite{BerChe18,Bhattacharya2020}.
Note that this is in sharp contrast to $O(m)$ for the offline variants of these problems, which can be solved by breadth-first or depth-first search.
More generally, research suggests there are computational barriers to solving unconstrained reachability problems in incremental and dynamic graphs~\cite{abboud2014popular,fan2017incremental}.

\begin{figure}[t]
\centering
\begin{tikzpicture}
\node[state, open] (q1) {$1$};
\node[state, below left of=q1, accepting] (q2) {$2$};
\node[state, below right of=q1, open] (q3) {$3$};
\draw (q1) edge (q2);
\draw (q1) edge (q3);
\end{tikzpicture}
\caption{
GID consisting of the sequence of updates
$\code{E}(1, 2)$, $\code{E}(1, 3)$, $\code{T}(2)$.
Terminal states are drawn as double circles.
After the update $\code{T}(2)$, states $1$ and $2$ are known to be live.
State $3$ is not dead in this GID, as a future update may cause it to be live.
}
\label{ex:gid-1}
\end{figure}

\begin{figure}[t]
\centering
\begin{tikzpicture}
\node[state, open] (q1) {$1$};
\node[state, below left of=q1, accepting] (q2) {$2$};
\node[state, below right of=q1, open] (q3) {$3$};
\node[state, above right of=q3] (q4) {$4$};
\node[state, below right of=q4] (q5) {$5$};
\draw (q1) edge (q2);
\draw (q1) edge (q3);
\draw (q4) edge (q3);
\draw (q4) edge (q5);
\end{tikzpicture}
\caption{
GID extending \Cref{ex:gid-1} with additional updates $\code{E}(4, 3)$, $\code{E}(4, 5)$, $\code{C}(4)$, $\code{C}(5)$.
Closed states are drawn as solid circles.
After the update $\code{C}(5)$ (but not earlier), state $5$ is dead.
State $4$ is not dead because it can still reach state $3$.
}
\label{ex:gid-2}
\end{figure}

\paragraph{This paper.}
To improve on prior algorithms, our key observation is that in many applications
(including our motivating applications above),
edges are not added adversarially, but \emph{from one state at a time}
as the states are explored.
As a result, we know when a state will have no further outgoing edges.
This enables us to (i) identify dead states incrementally, rather than only after the whole state space is explored; and (ii) obtain more efficient algorithms than currently exist for general graph reachability.

We introduce \emph{guided incremental digraphs} (GIDs), a variation on incremental graphs.
Like an incremental directed graph, a guided incremental digraph may be updated
by adding new edges between states, or
a state may be labeled as \emph{closed}, meaning it will receive no further outgoing edges.
Some states are designated as \emph{terminal}, and we say that a state is \emph{live} if it can
reach a terminal state and \emph{dead} if it will never reach a terminal
state in any extension -- i.e. if all reachable states from it are closed (see \Cref{ex:gid-1,ex:gid-2}).
To our knowledge, the problem of detecting dead states in such a system has not been studied by existing work in graph algorithms.
Our problem can be solved through solving SCC maintenance, but not necessarily the other way around (\Cref{sec:graph}, \Cref{prop:scc-reduction}).
We provide two new algorithms for dead-state detection in GIDs.

First, we show that the dead-state detection problem for GIDs can be solved in time $O(m \cdot \log m)$ for $m$ edge additions, within a logarithmic factor of the $O(m)$ cost for offline search.
The worst-case performance of our algorithm thus strictly improves on the $O(m^{3/2})$ upper bound for SCC maintenance in general incremental graphs.
Our algorithm is technically sophisticated, and utilizes several data structures and existing results in online algorithms:
in particular, Union-Find~\cite{Tar75}
and Henzinger and King's Euler Tour Trees~\cite{henzinger1999randomized}.
The main idea is that, rather than explicitly computing the
set of SCCs, for closed states we maintain
a single path to a non-closed (open) state.
This turns out to reduce the problem to quickly determining whether two states are currently assigned a path to the same open state.
On the other hand, Euler Tour Trees can solve \emph{undirected} reachability for graphs that are forests in logarithmic time.\footnote{
    Reachability in dynamic forests can also be solved by Sleator-Tarjan trees~\cite{sleator1983data}, Frederickson's Topology Trees~\cite{frederickson1997data},
    or Top Trees~\cite{alstrup2005maintaining}.
    Of these, we found Euler Tour Trees the easiest to work with in our implementation.
    See also~\cite{tarjan2010dynamic}.
}
The challenge then lies in figuring out how to reduce directed connectivity in the graph of paths to an undirected forest connectivity problem.
At the same time, we must maintain this reduction under Union-Find state merges,
in order to deal with cycles that are found in the graph along the way.

While as theorists we would like to believe that asymptotic complexity is enough,
the truth is that the use of complex data structures (1) can be prohibitively
expensive in practice due to constant-factor overheads, and (2) can make algorithms substantially more difficult to implement, leading practitioners to prefer simpler approaches.
To address these needs,
in addition to the logarithmic-time algorithm,
we provide a second \emph{lazy} algorithm which avoids the user of Euler Tour Trees,
and only uses union-find.
This algorithm is based on an optimization of adding shortcut \emph{jump} edges for long paths in the graph to quickly determine reachability.
This approach aims to perform well in practice on typical graphs,
and is evaluated in our evaluation along with the logarithmic time algorithm,
though we do not prove its asymptotic complexity.

Finally, we implement and empirically evaluate both of our algorithms for GIDs against several baselines
in 5.5k lines of code in Rust~\cite{matsakis2014rust}.
Our evaluation focuses on the performance of the GID data structure itself,
rather than its end-to-end performance in applications.
To ensure an apples-to-apples comparison with existing approaches, we put particular focus on providing a directed graph data structure backend shared by all algorithms,
so that the cost of graph search as well as state and edge merges is identical across algorithms.
We implement two \naive{} baselines,
as well as an implementation of the state-of-the-art solution based on
maintaining SCCs, BFGT~\cite{BFGT15} in our framework.
To our knowledge, the latter is the first implementation of BFGT specifically for SCC maintenance.
On a collection of generated benchmark GIDs, random GIDs, and GIDs directly pulled from the
regex application, we demonstrate a substantial improvement over BFGT
for both of our algorithms.
For example, for larger GIDs (those with over $100$K updates),
we observe a $110$-$530$x speedup over BFGT.

\paragraph{Contributions.}
Our primary contributions are:
\begin{itemize}
\item
\emph{Guided incremental digraphs} (GIDs),
a formalization of incremental live and dead state detection
which supports labeling \emph{closed} states.
(\Cref{sec:graph})
\item
Two algorithms for the state classification problem in GIDs:
first, an algorithm that works in amortized $O(\log m)$ time per update, improving upon the state-of-the-art amortized $O(\sqrt{m})$ per update for incremental graphs;
and second, a simpler algorithm based on lazy heuristics. (\Cref{sec:graph-algorithm})
\item
An \visiblehref{https://github.com/cdstanford/gid}{open-source implementation} of GIDs in Rust, and an evaluation which demonstrates up to two orders of magnitude speedup over BFGT.
(\Cref{sec:eval})
\end{itemize}

Following the above, we expand on the application of GIDs to regex solving in SMT (\Cref{sec:regex}), survey related work (\Cref{sec:rw}), and conclude (\Cref{sec:conclusion}).
This is the extended version of the paper; it includes about 1.5 pages of additional proofs and formal details compared to the published version.

\section{Guided Incremental Digraphs}
\label{sec:graph}

\subsection{Problem Statement}
\label{subsec:graph-problem}

An incremental digraph is a sequence of edge updates $\code{E}(u, v)$, where the algorithmic challenge in this context is to produce some output after each edge is received (e.g., whether or not a cycle exists).
If the graph also contains updates $\code{T}(u)$ labeling a state as \emph{terminal}, then we say that a state is \emph{live} if it can reach a terminal state in the current graph.
In a \emph{guided} incremental digraph, we also include updates $\code{C}(u)$ labeling a state as \emph{closed}, meaning that will not receive any further outgoing edges.

\begin{definition}
Define a \emph{guided incremental digraph (GID)} to be a sequence of updates, where each update is one of the following:
\label{def:gid}
\begin{enumerate}[(i)]
\item a new directed \emph{edge} $\code{E}(u, v)$;
\item a label $\code{T}(u)$ which indicates that $u$ is \emph{terminal}; or
\item a label $\code{C}(u)$ which indicates that $u$ is \emph{closed}, i.e. no further edges will be added going out from $u$ (or labels to $u$).
\end{enumerate}
\end{definition}

The GID is
\emph{valid} if the \emph{closed} labels are correct: there are no instances of $\code{E}(u, v)$ or $\code{T}(u)$ after an update $\code{C}(u)$.
The \emph{denotation} of $G$ is the directed graph $(V, E)$ where $V$ is the set of all states $u$ which have occurred in any update in the sequence, and $E$ is the set of all $(u, v)$ such that $\code{E}(u, v)$ occurs in $G$.
An \emph{extension} of a valid GID $G$ is a valid GID $G'$ such that $G$ is a prefix of $G'$.
In a valid GID $G$, we say that a state $u$ is \emph{live} if there is a path from $u$ to a terminal state in the denotation of $G$; and a state $u$ is \emph{dead} if it is not live in \emph{any} extension of $G$.
Notice that in a GID without any $\code{C}(u)$ updates, no states are dead as an edge may be added in an extension which makes them live.

We provide an example of a valid GID
in \Cref{ex:gid-1,ex:gid-2}
consisting of the following sequence of updates:
\code{E}(1, 2), \code{E}(1, 3), \code{T}(2), \code{E}(4, 3), \code{E}(4, 5), \code{C}(4), \code{C}(5).
Terminal states $\code{T}(u)$ are drawn as double circles;
closed states, as single circles $\code{C}(u)$;
and states that are not closed, as dashed circles.

\begin{definition}
Given as input a valid GID,
the \emph{GID state classification problem}
is to output, in an online fashion after each update,
the set of new live and new dead states.
That is, output $\code{Live}(u)$ or $\code{Dead}(u)$ on the smallest prefix of updates such that $u$ is live or dead on that prefix, respectively.
\label{def:problem}
\end{definition}

\subsection{Existing Approaches}
\label{subsec:graph-existing}

In many applications, one might choose to classify
dead states offline, after the entire state space is enumerated.
This leads to a linear-time algorithm via either DFS or BFS, but it does not
solve our problem (\Cref{def:problem})
because it is not incremental.
\Naive{} application of this idea leads to $O(m)$ per update for $m$ updates ($O(m^2)$ total), as we may redo the entire search after each update.

For acyclic graphs, there exists an amortized $O(1)$-time per update algorithm for the problem
(\Cref{def:problem}):
maintain the graph as a list of forward- and backward-edges at each state.
When a state $v$ is marked terminal, do a DFS along backward-edges to determine
all states $u$ that can reach $v$ not already marked as live, and mark them live.
When a state $v$ is marked closed, visit all forward-edges from $v$; if all are dead,
mark $v$ as dead and recurse along all backward-edges from $v$.
As each edge is visited only when marking a state live or dead, it is only visited a constant number of times overall (though we may use more than $O(1)$ time on some particular update pass).
Additionally, the live state detection part of this procedure still works for graphs containing cycles.

The challenge, therefore, lies primarily in detecting dead states in graphs which may contain cycles.
For this, the breakthrough approach from~\cite{BFGT15}
maintains a \emph{condensed} graph which is acyclic,
where the vertices in the condensed graph represent strongly connected components (SCCs) of states.
The mapping from states to SCCs is maintained using a
Union-Find~\cite{Tar75} data structure.
Maintaining the condensed graph requires $O(\sqrt{m})$ time per update.
To avoid confusing closed and non-closed states,
we also have to make sure that they are not merged into the same SCC;
the easiest solution to this is to withhold all edges from each state $u$ in the graph
until $u$ are closed, which ensures that $u$ must be in a SCC on its own.
Once we have the condensed graph with these modifications,
the same algorithm as in the previous paragraph works
to identify live and dead states.
Since each edge is only visited when a state is marked closed or live,
each edge is visited only once throughout the algorithm,
we use only amortized $O(1)$ additional time to calculate live and dead
states.
While this SCC maintenance algorithm ignores the
fact that edges do not occur from closed states $\code{C}(u)$,
this still proves the following result:

\begin{proposition}
GID state classification
reduces to SCC maintenance.
That is, suppose we have an algorithm over incremental graphs
that maintains the set of SCCs in
$O(f(m, n))$ total time
given $n$ states and $m$ edge additions.\footnote{To be precise, ``maintains'' means that
(i) we can check whether two states are in the same SCC
in $O(1)$ time; and
(ii) we can iterate over all the states, edges from, or edges into a
SCC in $O(1)$ time per state or edge.}
Then there exists an algorithm to solve GID
state classification in $O(f(m, n))$ total
time.
\label{prop:scc-reduction}
\end{proposition}

Despite this reduction one way, there is no obvious reduction the other way --
from cycle detection or SCCs to \Cref{def:problem}.
This is because, while the existence of a cycle of non-live states implies bi-reachability between all states in the cycle, it does not necessarily imply that all of the bi-reachable states are dead.

\section{Algorithms}
\label{sec:graph-algorithm}

This section presents \Cref{alg:log}, which solves the state classification problem in logarithmic time (\Cref{thm:log}); and \Cref{alg:jump}, an alternative lazy solution.
Both algorithms are optimized versions of \Cref{alg:first-cut}, a first-cut algorithm which establishes the structure of our approach.
We begin by establishing some basic terminology shared by all of the algorithms (see \Cref{fig:state-classification}).

\begin{figure}[t]
\centering
\begin{tabular}{p{1.4cm}p{10.1cm}}
\toprule
Live & Some reachable state from $u$ is terminal. \\
Dead & All reachable states from $u$ (including $u$) are closed and not terminal. \\
Unknown & $u$ is closed, but not live or dead. \\
Open & $u$ is not live and not closed. \\
\midrule
Terminal & A state $u$ labeled by $\code{T}(u)$. \\
Closed & A state $u$ labeled by $\code{C}(u)$. \\
Canonical & A state $x$ such that $\code{UF.find}(x) = x$.
\\
$u, v, w$ & States (may or may not be canonical). \\
$x, y, z$ & Canonical states (i.e., states in the condensed graph). \\
Successor $\code{succ}(x)$ & For an unknown, canonical state $x$, a uniquely chosen $v$ such that $(x, v)$ is an edge, and following the path of successors leads to an open state. \\
\bottomrule
\end{tabular}

\caption{
  \emph{Top:} Basic classification of GID states into four disjoint categories.
  \emph{Bottom:} Additional terminology used in this paper.
}
\label{fig:state-classification}
\end{figure}

States in a GID can be usefully classified as exactly one of four \emph{statuses}: \emph{live}, \emph{dead}, \emph{unknown}, or \emph{open},
where \emph{unknown} means ``closed but not yet live or dead'',
and \emph{open} means ``not closed and not live''.
Note that a state may be live and neither open nor closed;
this terminology keeps the classification disjoint.
Pragmatically, for live states it does not matter if they are
classified as open or closed, since edges from those states no
longer have any effect.
However, all dead and unknown states are closed, and no states are both
open and closed.

Given this classification, the intuition is that for each unknown state $u$, we only need \emph{one} path from $u$ to an open
state to prove that it is not dead; we want to maintain one such path for all unknown states.
To maintain all of these paths simultaneously,
we maintain an acyclic directed \emph{forest} structure on
unknown and open states where the roots
are open states, and all non-root states have a single edge to another state, called its \emph{successor}.
Edges other than successor edges can be temporarily ignored, except for when marking live states; these are kept as \emph{reserve} edges.
Specifically, we add every edge $(u, v)$ as a backward-edge from $v$ (to allow propagating live states), but for edges not in the forest we keep $(u, v)$ in a reserve list from $u$.
We store all edges, including backward-edges, in the original order $(u, v)$.
The reserve list edge becomes relevant only when either (i) $u$ is marked as closed, or (ii) $u$'s successor is marked as dead.

In order to deal with cycles, we need to maintain the forest of unknown states not on the original graph, but on a union-find \emph{condensed graph}, similar to~\cite{Tar75}.
When we find a cycle of unknown states, we \emph{merge} all states in the cycle by calling the union method in the union-find.
We refer to a state as \emph{canonical} if it is the canonical representative of its equivalence class in the union find; the condensed graph is a forest on canonical states.
We use $x, y, z$ to denote canonical states (states in the condensed graph),
and $u, v, w$ to denote the original states (not known to be canonical).
Following~\cite{Tar75}, we maintain edges as linked lists
rather than sets, and
using the original states instead of canonical states;
this is important as it allows combining edge lists in $O(1)$ time
when merging states.

\begin{algorithm}[tp]
\small
\caption{First-cut algorithm.}
\label{alg:first-cut}
\begin{algorithmic}[1]
  \State \code{V}: a type for states (integers) (variables $u, v, \ldots$)
  \State \code{E}: the type of edges, equal to $(\code{V}, \code{V})$
  \State \code{UF}: a union-find data structure over \code{V}
  \State \code{X}: the set of canonical states in \code{UF} (variables $x, y, z, \ldots$)
  \State \code{status}: a map from \code{X} to \code{Live}, \code{Dead}, \code{Unknown}, or \code{Open}
  \State \code{succ}: a map from \code{X} to \code{V}
  \State \code{res} and \code{bck}: maps from \code{X} to linked lists of \code{E}
  \Procedure{OnEdge}{$\code{E}(u, v)$}
    \State $x \gets \code{UF.find}(u)$;
    $y \gets \code{UF.find}(v)$
    \If{$\code{status}(y) = \code{Live}$}
      \State \Call{OnTerminal}{$\code{T}(x)$} \Comment{mark $x$ and its ancestors live}
    \ElsIf{$\code{status}(x) \ne \code{Live}$} \Comment{$\code{status}(x)$ must be \code{Open}}
      \State append $(u, v)$ to $\code{res}(x)$
      \State append $(u, v)$ to $\code{bck}(y)$
    \EndIf
  \EndProcedure
  \Procedure{OnTerminal}{$\code{T}(v)$}
    \State $y \gets \code{UF.find}(v)$
    \ForAll{$x$ in DFS backwards (along $\code{bck}$) from $y$ not already $\code{Live}$}
      \State $\code{status}(x) \gets \code{Live}$
      \State \Output $\code{Live}(x')$ for all $x'$ in $\code{UF.iter}(x)$
    \EndFor
  \EndProcedure
  \Procedure{OnClosed}{$\code{C}(v)$}
    \State $y \gets \code{UF.find}(v)$
    \If{$\code{status}(y) \ne \code{Open}$}
      \Return \Comment{$y$ is already live or closed}
    \EndIf
    \While{$\code{res}(y)$ is nonempty}
      \State pop $(v, w)$ from $\code{res}(y)$;
      $z \gets \code{UF.find}(w)$
      \If{$\code{status}(z) = \code{Dead}$} \Continue
      \ElsIf{$\Call{CheckCycle}{y, z}$}
        \ForAll{$z'$ in cycle from $z$ to $y$}
          $z \gets \Call{Merge}{z, z'}$
        \EndFor
      \Else
        \State $\code{status}(y) \gets \code{Unknown}$;
          $\code{succ}(y) \gets z$;
        \State \Return
      \EndIf
    \EndWhile
    \State $\code{status}(y) \gets \code{Dead}$;
      \Output $\code{Dead}(y')$ for all $y'$ in $\code{UF.iter}(y)$
    \State $\code{ToRecurse} \gets \varnothing$
    \ForAll{$(u, v)$ in $\code{bck}(y)$}
      \State $x \gets \code{UF.find}(u)$
      \If{$\code{status}(x) = \code{Unknown}$ and $\code{UF.find}(\code{succ}(x)) = y$}
        \State $\code{status}(x) \gets \code{Open}$ \Comment{temporary -- marked closed on recursive call}
        \State add $x$ to \code{ToRecurse}
      \EndIf
    \EndFor
    \ForAll{$x$ in \code{ToRecurse}}
      \Call{OnClosed}{$\code{C}(x)$}
    \EndFor
  \EndProcedure
  \Procedure{CheckCycle}{$y, z$} \Returning{\code{bool}}
    \While{$\code{status}(z) = \code{Unknown}$}
      $z \gets \code{UF.find}(\code{succ}(z))$ \Comment{get root state from $z$}
    \EndWhile
    \State \Return $y = z$
  \EndProcedure
  \Procedure{Merge}{$x, y$} \Returning{\code{V}}
    \State $z \gets \text{UF.union}(x, y)$
    \State $\code{bck}(z) \gets \code{bck}(x) + \code{bck}(y)$ \Comment{$O(1)$ linked list append}
    \State $\code{res}(z) \gets \code{res}(x) + \code{res}(y)$ \Comment{$O(1)$ linked list append}
    \State \Return $z$
  \EndProcedure
\end{algorithmic}
\end{algorithm}

\subsection{First-Cut Algorithm}

\Cref{alg:first-cut} is a first cut based on these ideas.
The procedures \Call{OnEdge}{} and \Call{OnTerminal}{} contain all the logic to identify live states, using $\code{bck}$ to look up backward-edges;
\Call{OnTerminal}{} doubles as a ``mark live'' function when it is called by \Call{OnEdge}{}.
The procedure \Call{OnClosed}{} tries to assign a successor edge to a newly closed state,
to prove that it is not dead.
In case we run out of reserve edges, the state is marked dead and we recursively call \Call{OnClosed}{} along backward-edges, which will either set a new successor or mark them dead.

The union-find data structure $\code{UF}$ provides $\code{UF.union}(v_1, v_2)$, $\code{UF.find}(v)$, and $\code{UF.iter}(v)$:
$\code{UF.union}$ merges $v_1$ and $v_2$ to refer to the same canonical state, $\code{UF.find}$ returns the canonical state for $v$,
and $\code{UF.iter}$ iterates over states equivalent to $v$.
These use amortized $\alpha(n)$ for $n$ updates, where $\alpha(n) \in o(\log n)$ is the inverse Ackermann function.
We only merge states if they are bi-reachable from each other, and both unknown;
this implies that all states equivalent to a state $x$ have the same status.
Each edge $(u, v)$ is always stored in the maps $\code{res}$ and $\code{bck}$ using its original states (i.e., edge labels are not updated when states are merged);
but we can quickly obtain the corresponding edge on canonical states
via $(\code{UF.find}(u), \code{UF.find}(v))$.
Once a state is marked \code{Live} or \code{Dead}, its edge maps are no longer used.

\paragraph{Invariants.}
Altogether, we respect the following invariants.
\emph{Successor} and \emph{no cycles} describe the forest structure,
and, \emph{edge representation} ensures that all edges
in the input GID are represented somehow in the current graph.
\begin{itemize}
\item \emph{Merge equivalence:} For all states $u$ and $v$, if $\code{UF.find}(u) = \code{UF.find}(v)$, then $u$ and $v$ are bi-reachable and both closed. (This implies that $u$ and $v$ are both live, both dead, or both unknown.)
\item \emph{Status correctness:} For all $u$, $\code{status}(\code{UF.find}(u))$ equals the status of $u$.
\item \emph{Successor edges:} If $x$ is unknown, then $\code{succ}(x)$ is defined and is an unknown or open state. If $x$ is open, then $\code{succ}(x)$ is not defined.
\item \emph{No cycles:} There are no cycles among the set of edges $(x, \code{UF.find}(\code{succ}(x)))$, over all unknown and open canonical states $x$.
\item \emph{Edge representation:}
For all edges $(u, v)$ in the input GID, at least one of the following holds:
(i) $(u, v) \in \code{res}(\code{UF.find}(v))$;
(ii) $v = \code{succ}(\code{UF.find}(u))$;
(iii) $\code{UF.find}(u) = \code{UF.find}(v)$;
(iv) $u$ is live; or
(v) $v$ is dead.
\end{itemize}

\begin{theorem}
\Cref{alg:first-cut} is correct.
\label{prop:first-cut-correct}
\end{theorem}

\begin{proof}
The \emph{status correctness} invariant implies correct output at each step,
so it suffices to argue that all of the invariants above are preserved.
The hardest case is \Call{OnClosed}{$\code{C}(u)$}.
The procedure is recursive;
on recursive calls,
some states are \emph{temporarily} marked $\code{Open}$,
meaning they are roots in the forest structure.
During these recursive calls, we need a slightly generalized invariant:
each forest root corresponds to a pending call to \Call{OnClosed}{$\code{C}(u)$} (i.e., an element of \code{ToRecurse} for some call on the stack)
and is a state that is dead iff all of its reserve edges are dead.
This means that the $\code{status}$ labels for unknown- and open-labeled states are not necessarily correct during recursive calls; we are only concerned with preserving the forest structure, so that they will be correct after all calls complete.

Upon receiving $\code{E}(u, v)$ or $\code{T}(u)$, some dead, unknown, or open states may become live, but this does not change the status of any other states. As $\code{bck}$ stores all backward-edges (not just $\code{succ}$ edges),
live states are marked correctly.
This preserves the forest invariants because if an unknown or open state
is marked live, so are all its predecessors.
This also preserves the edge representation invariant, either by adding new edges to case (i) and case (iv) (for \code{E} updates) or \emph{moving} edges from cases (i)-(iii) to case (iv) (for both \code{E} and \code{T} updates).

Upon receiving $\code{C}(u)$, without loss of generality we may assume $\code{status}(u) = \code{Open}$ (the opposite only occurs if the input GID has duplicate $\code{C}(u)$ updates, or if $\code{C}(u)$ occurs for a live state).
At the top of the while loop, there are three cases:
\begin{itemize}
\item If a reserve edge is available, and its target is dead, then we discard it. This preserves edge representation because that edge still satisfies (iv).
\item If a reserve edge is available and its target is not dead, then we see if adding that edge creates a cycle by calling \code{CheckCycle}.
If no, then the two trees are distinct, so we add an edge between them.
This preserves edge representation by moving that edge from case (i) to case (ii).
If yes, then we collapse the cycle in the forest to a single state by repeated calls to $\Call{Merge}{x, y}$.
This preserves edge representation by moving the edges in the cycle from case (ii) to case (iii).
The forest structure of the \code{succ} is also preserved because if a graph is a forest, then it remains a forest after merging adjacent states.
\item Finally, if there are no reserve edges left, then because of edge representation, and because there is no successor from $u$, all edges from $u$ must lead to dead states (case (v)) and therefore $u$ is dead.
This case splits the tree rooted at $u$ into one tree for each of its predecessors,
preserving the forest structure.
Each of the \code{succ} edges from predecessors are deleted; this preserves edge representation by moving those edges from case (ii) to case (v).
\end{itemize}
In any case, the generalized invariant for recursive calls is preserved.
\qed
\end{proof}

\paragraph{Complexity.}
The core inefficiency in \Cref{alg:first-cut} --- what we need to improve --- lies in \Call{CheckCycle}{}.
The procedure repeatedly sets $z \gets \code{succ}(z)$ to find the tree root, which in general could be linear time in the number of edges.
For example, this inefficiency results in $O(m^2)$ work for a linear graph read in backwards order:
\code{E}(2, 1), \code{C}(2), \code{E}(3, 2), \code{C}(3), \ldots, \code{E}(n, n-1), \code{C}(n).

All other procedures use amortized $\alpha(m)$
time per update for $m$ updates, using array lists to represent the maps \code{fwd}, \code{bck}, and \code{succ} for $O(1)$ lookups.
To do the amortized analysis, the cost of each call to \Call{OnClosed}{} can be assigned \emph{either} to the target of an edge being marked dead, \emph{or} to an edge being merged as part of a cycle,
and both of these events can only happen once per edge added to the GID.
And the \Call{OnTerminal}{} calls and loop iterations only run once per edge in the graph when the target of that edge is marked live or terminal.

\begin{algorithm}[tp]
\small
\begin{algorithmic}[1]
  \State All data from \Cref{alg:first-cut};
    \code{succ}: a map from \code{X} to \code{E} (instead of to \code{V})
  \State \code{EF}: Euler Tour Trees data structure providing:
    \code{EF.add}, \code{EF.remove}, \code{EF.connected}
  \Procedure{OnEdge}{},
    \Call{Merge}{}
    as in \Cref{alg:first-cut}
  \EndProcedure
  \Procedure{OnTerminal}{$\code{T}(v)$}
    \State $y \gets \code{UF.find}(v)$
    \ForAll{$x$ in DFS backwards (along $\code{bck}$) from $y$ not already $\code{Live}$}
    \If{$\code{status}(x) = \code{Unknown}$}
    \State \Comment{The following line is not strictly necessary, but simplifies the analysis}
    \State $(u, v) \gets \code{succ}(x)$;
        delete $\code{succ}(x)$;
        $\code{EF.remove}(u, v)$
    \EndIf
    \State $\code{status}(x) \gets \code{Live}$;
      \Output $\code{Live}(x')$ for all $x'$ in $\code{UF.iter}(x)$
    \EndFor
  \EndProcedure
  \Procedure{OnClosed}{$\code{C}(v)$}
    \State $y \gets \code{UF.find}(v)$
    \If{$\code{status}(y) \ne \code{Open}$}
      \Return
    \EndIf
    \While{$\code{res}(y)$ is nonempty}
      \State pop $(v, w)$ from $\code{res}(y)$;
      $z \gets \code{UF.find}(w)$
      \If{$\code{status}(z) = \code{Dead}$} \Continue
      \ElsIf{$\Call{CheckCycle}{y, z}$}
        \ForAll{$z'$ in cycle from $z$ to $y$}
          $z \gets \Call{Merge}{z, z'}$
        \EndFor
      \Else
        \State $\code{status}(x) \gets \code{Unknown}$;
          $\code{succ}(x) \gets (v, w)$
        \State $\code{EF.add}(v, w)$;
          \Return
          \Comment{undirected edge; use original labels (not $(x, y)$)}
      \EndIf
    \EndWhile
    \State $\code{status}(y) \gets \code{Dead}$;
      $\code{ToRec} \gets \varnothing$;
      \Output $\code{Dead}(y')$ for all $y'$ in $\code{UF.iter}(y)$
    \ForAll{$(u, v)$ in $\code{bck}(y)$}
      \State $x \gets \code{UF.find}(u)$
      \If{$\code{status}(x) = \code{Unknown}$}
        \State $(u', v') \gets \code{succ}(x)$
        \If{$\code{UF.find}(v') = y$}
          \State $\code{EF.remove}(u', v')$;
            $\code{status}(x) \gets \code{Open}$;
            delete $\code{succ}(x)$;
            add $x$ to \code{ToRec}
        \EndIf
      \EndIf
    \EndFor
    \ForAll{$x$ in \code{ToRec}}
      \Call{OnClosed}{$\code{C}(x)$}
    \EndFor
  \EndProcedure
  \Procedure{CheckCycle}{$y, z$} \Returning{\code{bool}}
    \State \Return $\code{EF.connected}(y, z)$
  \EndProcedure
\end{algorithmic}
\caption{Logarithmic time algorithm.}
\label{alg:log}
\end{algorithm}

\subsection{Logarithmic Algorithm}

At its core, \Call{CheckCycle}{} requires solving an \emph{undirected} reachability
problem on a graph that is restricted to a forest.
However, the forest is changed not just by edge additions, but edge additions \emph{and} deletions.
While undirected reachability and reachability in directed graphs are both difficult to solve incrementally,
reachability in \emph{dynamic forests} can be solved in $O(\log m)$ time per operation.
This is the main intuition for our solution, using an Euler Tour Trees data structure \code{EF} of Henzinger and King~\cite{henzinger1999randomized}, shown in \Cref{alg:log}.

Unfortunately, this idea does not work straightforwardly -- once again because of the presence of cycles in the original graph.
We cannot simply store the forest as a condensed graph with edges on condensed states.
As we saw in \Cref{alg:first-cut}, it was important to store successor edges as edges into \code{V}, rather than edges into \code{X} -- this is the only way that we can merge states in $O(1)$, without actually inspecting the edge lists. If we needed to update the forest edges to be in \code{X}, this could require $O(m)$ work to merge two $O(m)$-sized edge lists as each edge might need to be relabeled in the \code{EF} graph.

To solve this challenge, we instead store the \code{EF} data structure on the original states, rather than the condensed graph; but we ensure that \emph{each canonical state is represented by a tree of original states}.
When adding edges between canonical states, we need to make sure to remember the original label $(u, v)$, so that we can later remove it using the original labels (this happens when its target becomes dead).
When an edge would create a cycle, we instead simply ignore it in the \code{EF} graph, because a line of connected trees forms a tree.

\paragraph{Summary and invariants.}
In summary, the algorithm reuses the data, procedures, and invariants from \Cref{alg:first-cut}, with the following important changes: (1) We maintain the EF data structure \code{EF}, a forest on \code{V}. (2) The successor edges are stored as their original edge labels $(u, v)$, rather than just as a target state. (3) The procedure \Call{OnClosed}{} is rewritten to maintain the graph \code{EF}. (4) The \emph{successor edges} and \emph{no cycles} invariants use the new \code{succ} representation: that is, they are constraints on the edges $(x, \code{UF.find}(v))$, where
$\code{succ}(x) = (u, v)$.
(5) We add the following two constraints on edges in \code{EF},
depending on whether those states are equivalent in the union-find structure.
\begin{itemize}
\item \emph{EF inter-edges:} For all \emph{inequivalent} $u, v$, $(u, v)$ is in the EF if and only if $(u, v) = \code{succ}(\code{UF.find}(u))$ or $(v, u) = \code{succ}(\code{UF.find}(v))$.
\item \emph{EF intra-edges:} For all unknown canonical states $x$,
the set of edges $(u, v)$ in the EF between states belonging to $x$ forms a tree.
\end{itemize}

\begin{theorem}
\Cref{alg:log} is correct.
\label{thm:log-correct}
\end{theorem}

\begin{proof}
Observe that the \code{EF} inter-edges constraint implies that \code{EF} only contains edges between unknown and open states, together with isolated trees.
In the modified \Call{OnTerminal}{} procedure, when marking states as live we remove
inter-edges, so we preserve this invariant.

Next we argue that given the invariants about \code{EF}, for an \emph{open} state $y$ the \Call{CheckCycle}{}
procedure returns true if and only if $(y, z)$ would create a directed cycle.
If there is a cycle of canonical states, then because canonical states are connected
trees in \code{EF}, the cycle can be lifted to a cycle on original states,
so $y$ and $z$ must already be connected in this cycle without the edge $(y, z)$.
Conversely, if $y$ and $z$ are connected in \code{EF}, then there is a path from $y$
to $z$, and this can be projected to a path on canonical states.
However, because $y$ is open, it is a root in the successor forest, so any path from $y$
along successor edges travels only on backward-edges; hence $z$ is an ancestor of $y$
in the \emph{directed} graph, and thus $(y, z)$ creates a directed cycle.

This leaves the \Call{OnClosed}{} procedure.
Other than the \code{EF} lines, the structure is the same as in
\Cref{alg:first-cut}, so the previous invariants are still preserved,
and it remains to check the \code{EF} invariants.
When we delete the successor edge and temporarily mark $\code{status}(x) = \code{Open}$ for recursive calls,
we also remove it from \code{EF},
preserving the inter-edge invariant.
Similarly, when we add a successor edge to $x$, we add it to \code{EF},
preserving the inter-edge invariant.
So it remains to consider when the set of canonical states changes,
which is when merging states in a cycle.
Here, a line of canonical states is merged into a single state,
and a line of connected trees is still a tree, so the intra-edge invariant still holds
for the new canonical state,
and we are done.
\qed
\end{proof}

\begin{theorem}
\Cref{alg:log} uses amortized logarithmic time per edge update.
\label{thm:log}
\end{theorem}

\begin{proof}
By the analysis of \Cref{alg:first-cut}, each line of the algorithm is executed
$O(m)$ times and there are $O(m)$ calls to \Call{CheckCycle}{}.
Each line of code is either constant-time, $\alpha(m) = o(\log m)$ time for the \code{UF} calls,
or $O(\log m)$ time for the \code{EF} calls,
so in total the algorithm takes $O(m \log m)$ time total, or amortized $O(\log m)$ time per edge.
\qed
\end{proof}

\begin{algorithm}[t]
\small
\begin{algorithmic}[1]
  \State All data from \Cref{alg:first-cut};
    \code{jumps}: a map from \code{X} to lists of \code{V}
  \Procedure{OnEdge}{},
    \Call{OnTerminal}{}
    \Call{OnClosed}{}
    as in \Cref{alg:first-cut}
  \EndProcedure
  \Procedure{CheckCycle}{$y, z$} \Returning{\code{bool}}
    \State \Return $y = \Call{GetRoot}{z}$
  \EndProcedure
  \Procedure{GetRoot}{$z$} \Returning{\code{V}}
    \If{$\code{status}(z) = \code{Open}$}
      \Return $z$
    \EndIf
    \If{$\code{jumps}(z)$ is empty}
      push $\code{succ}(z)$ to $\code{jumps}(z)$
        \Comment{set $0$th jump}
    \EndIf
    \Repeat{}
      pop $w$ from $\code{jumps}(z)$; $z' = \code{UF.find}(w)$
      \Comment{remove dead jumps}
    \Until{$\code{status}(z') \ne \code{Dead}$}
    \State push $z'$ to $\code{jumps}(z)$;
      $\text{result} \gets \Call{GetRoot}{z'}$
    \State $n \gets \code{length}(\code{jumps}(z))$;
      $n' \gets \code{length}(\code{jumps}(z'))$
    \If{$n \le n'$}
      push $\code{jumps}(z')[n - 1]$ to $\code{jumps}(z)$
        \Comment{set $n$th jump}
    \EndIf
    \State \Return \text{result}
  \EndProcedure
  \Procedure{Merge}{$x, y$} \Returning{\code{V}}
    \State $z \gets \text{UF.union}(x, y)$
    \State $\code{bck}(z) \gets \code{bck}(x) + \code{bck}(y)$;
      $\code{res}(z) \gets \code{res}(x) + \code{res}(y)$
    \State $\code{jumps}(z) \gets \code{empty}$;
      \Return $z$
  \EndProcedure
\end{algorithmic}
\caption{Lazy algorithm.}
\label{alg:jump}
\end{algorithm}

\subsection{Lazy Algorithm}

While the asymptotic complexity of $\log m$ could be the end of the story,
in practice, we found the cost of the \code{EF} calls to be a significant overhead.
The technical details of Euler Tour Trees include building an AVL-tree cycle for each tree, where the cycle contains each state of the graph once and each edge in the graph twice.
While this is elegant, it turns out that adding \emph{one edge} to \code{EF} results in
no less than \emph{seven} modifications to the AVL tree: a split at the source, then a split at the target, then an edge addition in both directions $(u, v)$ and $(v, u)$ to the cycle, and finally the four resulting trees need to be glued together (using three merge operations).\footnote{Our implementation actually uses nine modifications, as the splits at the source and target also disconnect the source and target states.}
Each one of these operations comes with a rebalancing operation which could do $\Omega(\log m)$ tree rotations and pointer dereferences to visit the nodes in the AVL tree.
Some optimizations may be possible -- including, e.g., combining rebalancing operations or considering variants of AVL trees with better cache locality. Nonetheless, these constant-factor overheads constitute a serious practical drawback for \Cref{alg:log}.

To address this, in this section, we investigate a simpler, lazy algorithm which avoids \code{EF}
and directly optimizes \Cref{alg:first-cut}.
For this, one idea in the right direction
is to store for each state a direct pointer
to the root
which results from repeatedly calling $\code{succ}$. But there are two issues with this.
First, maintaining this may be difficult (when the root changes, potentially updating a linear number of root pointers).
Second, the root may be marked dead, in which case we have to re-compute all pointers to that root.

Instead, we introduce a \emph{jump list} from each state:
intuitively, it will contain states after calling successor once, twice, four times, eight times, and so on at powers of two; and it will be updated lazily, at most once for every visit to the state.
When a jump becomes obsolete (the target dead), we just pop off the largest jump, so we do not lose all of our work in building the list.
We maintain the following additional information:
for each unknown canonical state $x$, a nonempty list of
\emph{jumps} $[v_0, v_1, v_2, \ldots, v_k]$, such that $v_0$ is reachable from $x$, $v_1$ is reachable from $v_0$, $v_2$ is reachable from $v_1$, and so on, and $v_1 = \code{succ}(x)$.

The resulting algorithm is shown in \Cref{alg:jump}.
The key procedure is \Call{GetRoot}{$z$}, which is called when adding a reserve edge $(y, z)$ to the graph.
In addition to all invariants from \Cref{alg:first-cut}, we maintain the following invariants
for \emph{every} unknown canonical state $x$,
where $\code{jumps}(x)$ is a list of states $v_0, v_1, v_2, \ldots, v_k$.
\emph{First jump:} if the jump list is nonempty, then $v_0 = \code{succ}(v)$.
\emph{Reachability:} $v_{i+1}$ is reachable from $v_i$ for all $i$.
The jump list also satisfies the following \emph{powers of two} invariant: on the path of canonical states from $v_0$ to $v_i$, the total number of states (including all states in each equivalence class) is at least $2^i$.
While this invariant is not necessary for correctness,
it is the key to the algorithm's practical efficiency:
it follows from this that \emph{if} the jump list is fully saturated for every state,
querying \Call{GetRoot}{$z$} will take only logarithmic time.
However, since jump lists are updated lazily,
the jump list may not be saturated, so
this does not establish a true asymptotic complexity for the algorithm.

\begin{theorem}
\Cref{alg:jump} is correct.
\label{thm:jump-correct}
\end{theorem}

\begin{proof}
The \emph{first jump} and \emph{reachability} invariants imply that $v_1, v_2, \ldots$ is some sublist of the states along the path from
an unknown state to its root, potentially followed by some dead states.
We need to argue that the subprocedure \Call{GetRoot}{} (i) receives
the same verdict as repeatedly calling $\code{succ}$ to find a cycle in the
first-cut algorithm and (ii) preserve both invariants.
For \emph{first jump}, if the jump list is empty, then \Call{GetRoot}{} ensures that the first jump is set to the successor state.
For \emph{reachability}, popping dead states from the jump list clearly preserves the invariant,
as does adding on a state along the path to the root, which is done when
$k' \ge k$.
Merging states preserves both invariants trivially
because we throw the jump list away,
and marking states live preserves both invariants trivially since
the jump list is only maintained and used for unknown states.
\qed
\end{proof}

\section{Experimental Evaluation}
\label{sec:eval}

The primary goal of our evaluation has been to experimentally validate the
performance of GIDs as a data structure in isolation, rather than their use in a particular
application.
Our evaluation seeks to answer the following questions:
\begin{description}
\item[Q1] How does our approach (\Cref{alg:log,alg:jump}) compare to the state-of-the-art approach based on maintaining SCCs?
\item[Q2] How does the performance of the studied algorithms vary when the class of input graphs changes (e.g., sparse vs. dense, structured vs. random)?
\item[Q3] Finally, how do the studied algorithms perform on GIDs taken from the example application to regexes described in \Cref{sec:regex}?
\end{description}

To answer \textbf{Q1}, we put substantial implementation effort into a common framework on which a fair comparison could be made between different approaches.
To this end, we implemented GIDs as a data structure in Rust which includes a graph data structure on top of which all algorithms are built. In particular, this equalizes performance across algorithms for the following baseline operations: state and edge addition and retrieval, DFS and BFS search, edge iteration, and state merging.
We chose Rust for our implementation for its performance, and because there does not appear to be an existing publicly available implementation of BFGT in any other language.\footnote{
  That is, BFGT for SCC maintenance.
  BFGT for cycle detection has been implemented before, for instance, in~\cite{GraphsJL} and formally verified in~\cite{gueneau2019formal}.
}
The number of lines of code used to implement these various structures is summarized in \Cref
{fig:eval-loc}.
We implement \Cref{alg:log,alg:jump} and compare them with the following baselines:

\begin{figure}[p]
  \begin{tabular}{lr}
    \toprule
    \makecell[l]{Implementation \\ Component} & LoC \\
    \midrule
    Common Framework & 1197 \\
    \Naive{} Algorithm & 78 \\
    Simple Algorithm & 98 \\
    BFGT Algorithm & 265 \\
    \Cref{alg:log} (Ours) & 253 \\
    \Cref{alg:jump} (Ours) & 283 \\
    Euler Tour Trees & 1510 \\
    Experimental Scripts & 556 \\
    Separated Unit Tests & 800 \\
    Utility & 217 \\
    Other & 69 \\
    Total & 5326 \\
    \bottomrule
  \end{tabular}
  \hspace{1cm}
  \begin{tabular}{llrr}
    \toprule
    Category & Benchmark & Source & Qty \\
    \midrule
    Basic
      & Line && 24 \\ %
      & Cycle && 24 \\ %
      & Complete && 18 \\ %
      & Bipartite && 14 \\ %
      & Total && 80 \\
    \midrule
    Random
      & Sparse && 260 \\ %
      & Dense && 130 \\ %
      & Total && 390 \\
    \midrule
    Regex
      & RegExLib~\cite{SMT12-regex} & 2061 & 37 \\
      & Handwritten~\cite{stanfordsymbolic} & 70 & 26 \\
      & Additional && 11 \\ %
      & Total && 74 \\
    \bottomrule
  \end{tabular}

\caption{
  \emph{Left:} Lines of code for each algorithm and other implementation components.
  \emph{Right:} Benchmark GIDs used in our evaluation. Where present, the source column
  indicates the quantity prior to filtering out trivially small graphs.
}
\label{fig:eval-loc}
\end{figure}

\begin{figure}[p]
\vspace{-0.3cm}

\begin{minipage}{.42\textwidth}
\includegraphics[width=\textwidth]{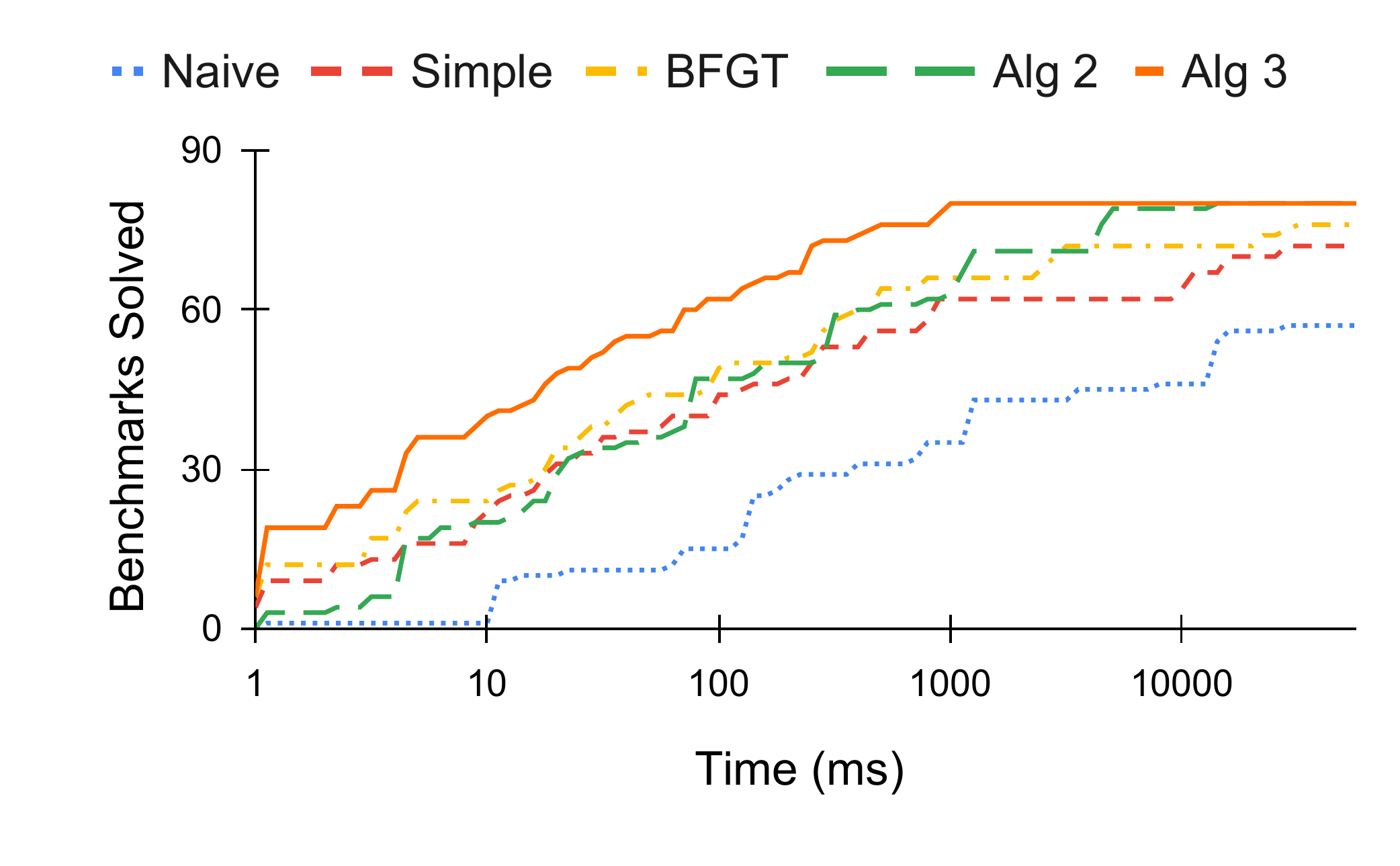}

\includegraphics[width=\textwidth]{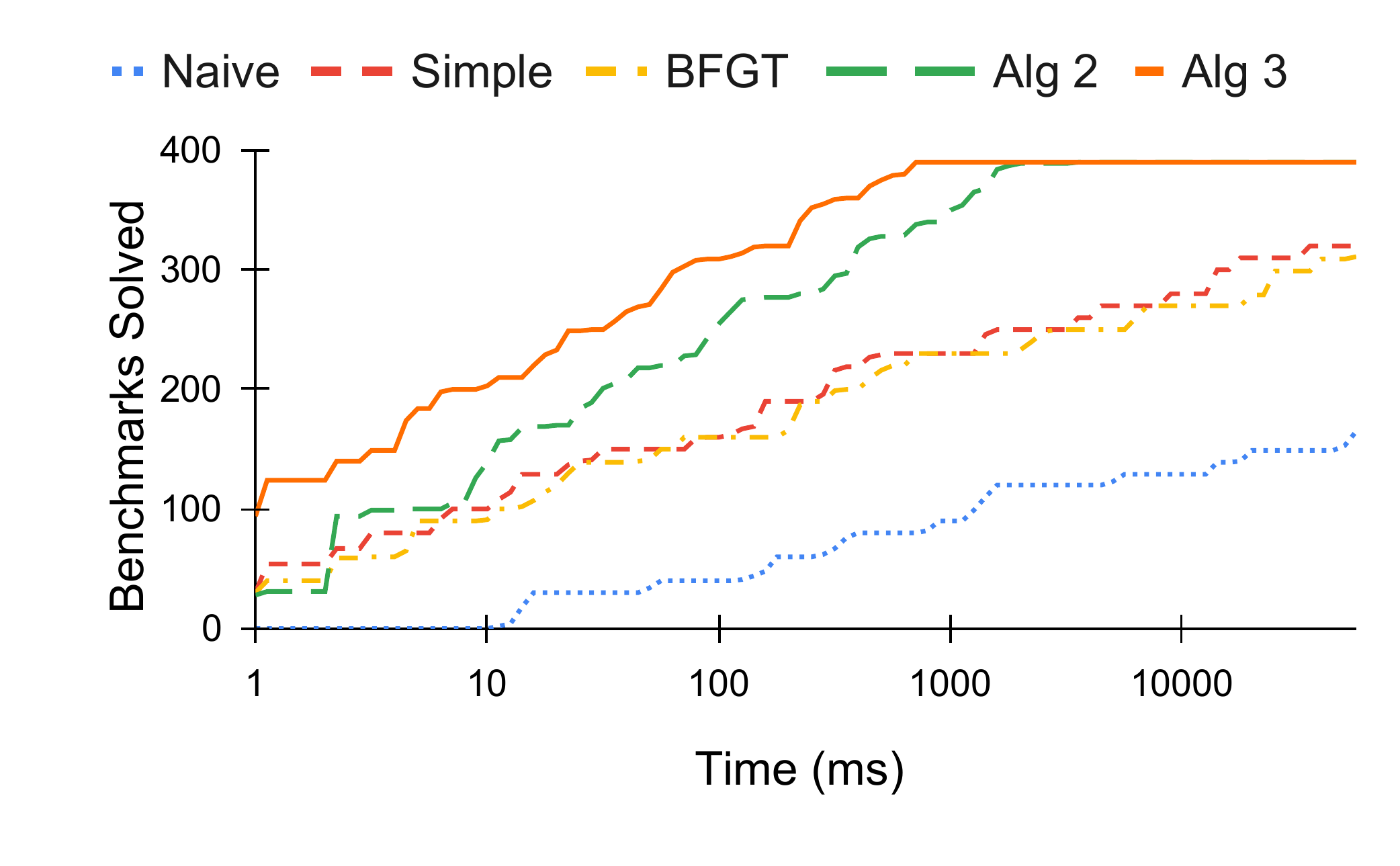}

\includegraphics[width=\textwidth]{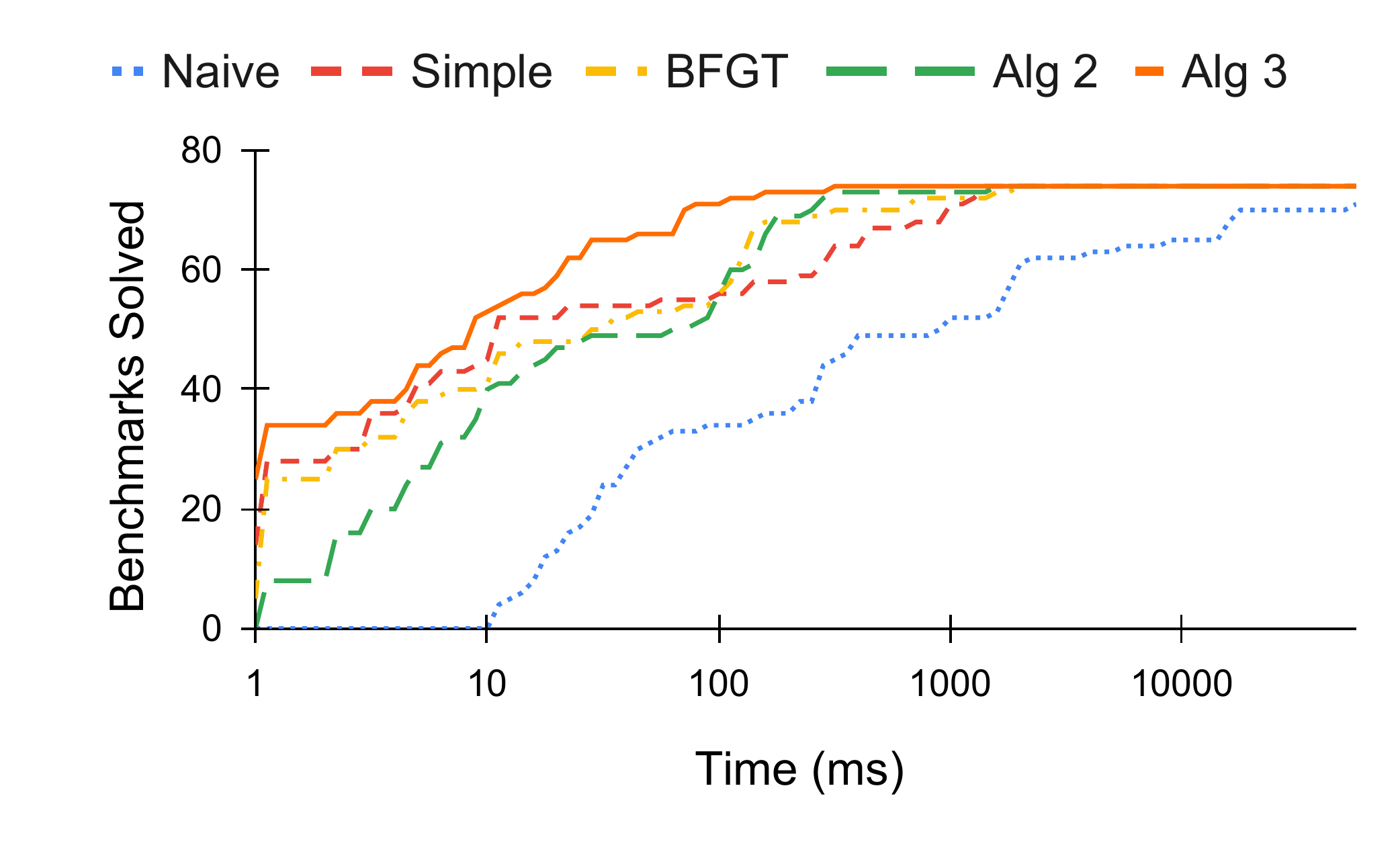}
\end{minipage}%
\begin{minipage}{.55\textwidth}
\includegraphics[width=\textwidth]{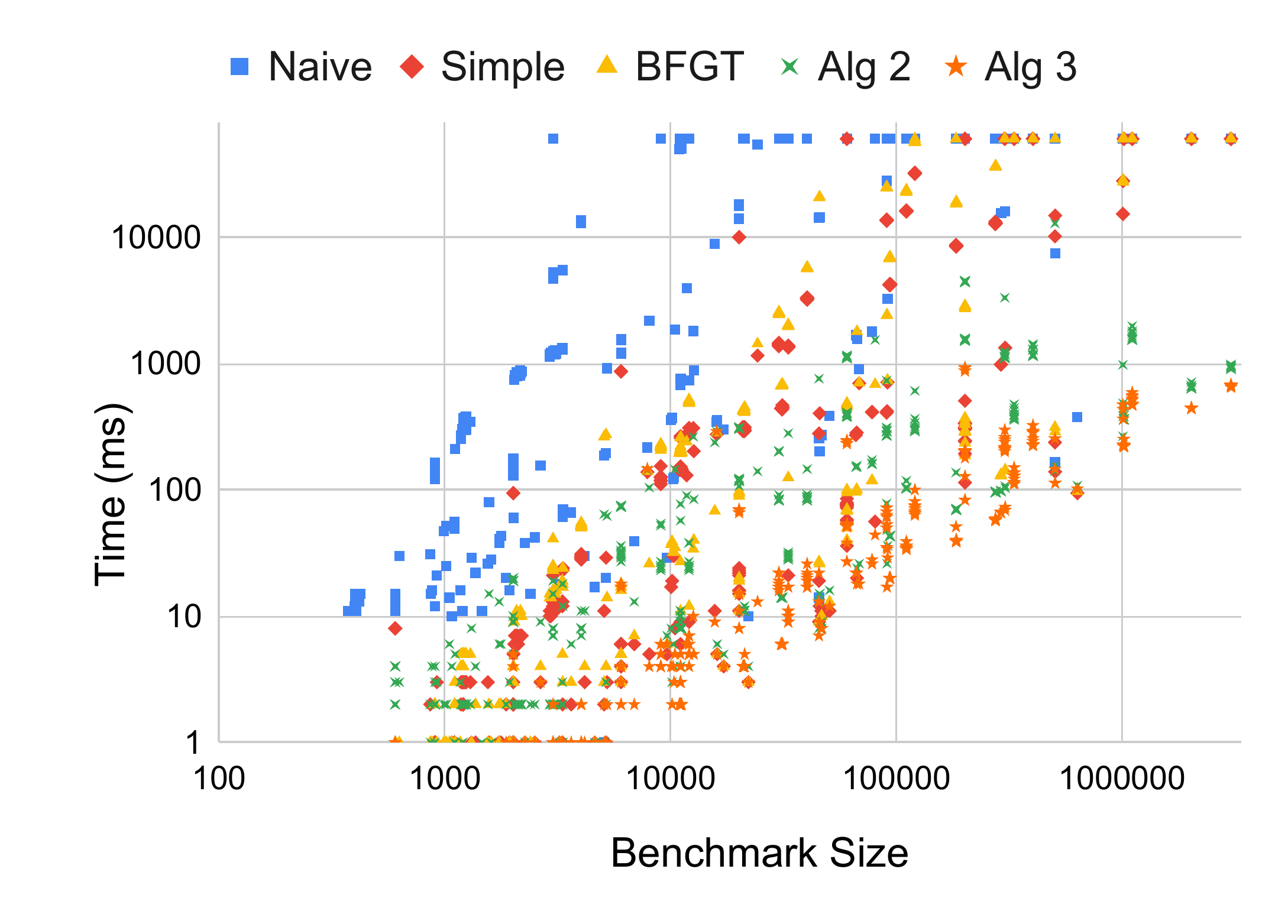}

\includegraphics[width=\textwidth]{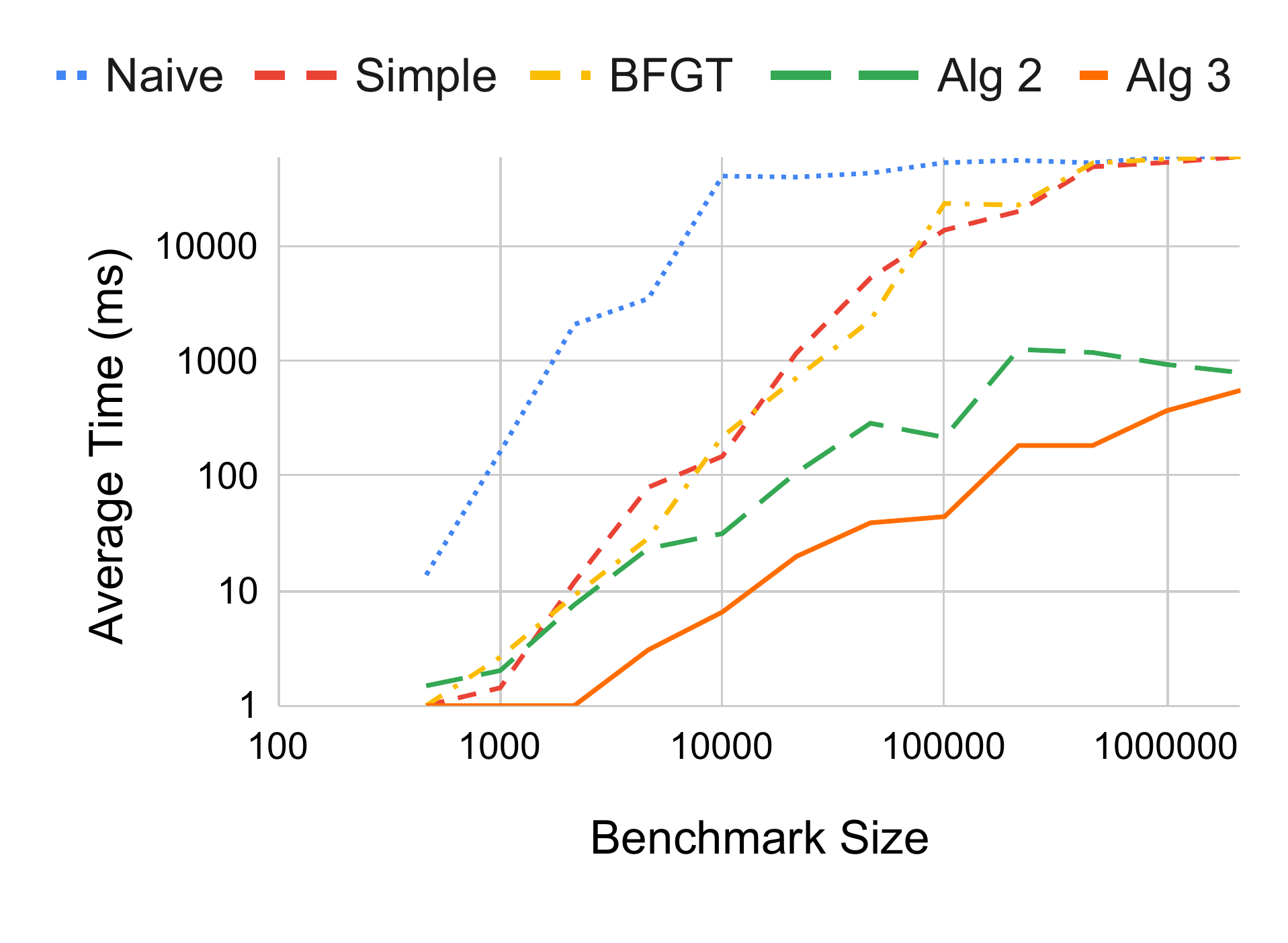}
\end{minipage}

\vspace{-0.3cm}
\caption{
  Evaluation results.
  \emph{Left:}
  Cumulative plot showing the number of benchmarks solved in time $t$ or less for basic GID classes (top), randomly generated GIDs (middle), and regex-derived GIDs (bottom).
  \emph{Top right:}
  Scatter plot showing the size of each benchmark vs time to solve.
  \emph{Bottom right:}
  Average time to solve benchmarks of size closest to $s$, where values of $s$ are chosen in increments of $1/3$ on a log scale.
}
\label{fig:eval-results}
\end{figure}

\begin{description}
\item[BFGT] The state-of-the-art approach based on SCC maintenance, using worst-case amortized $O(\sqrt{m})$ time per update~\cite{BFGT15}.
\item[Simple] A simpler version of BFGT that uses a forward-DFS to search for cycles.
Like \Cref{alg:first-cut}, it can take $\Theta(m^2)$ in the worst case.
\item[\Naive{}] A greedy upper bound for all approaches which re-computes the entire set of dead states using a linear-time DFS after each update.
\end{description}

To answer \textbf{Q2}, first, we compiled a range of basic graph classes which are designed to expose edge case behavior in the algorithms, as well as randomly generated graphs.
We focus on graphs with no live states, as live states are treated similarly by all algorithms.
Most of the generated graphs come in $2 \times 2 = 4$ variants:
(i) the states are either read in a forwards- or backwards- order; and
(ii) they are either \emph{dead} graphs, where there are no open states at the end and so everything gets marked dead; or \emph{unknown} graphs, where there is a single open state at the end, so most states are unknown.
In the unknown case, it is sufficient to have one open state at the end, as many open states can be reduced to the case of a single open state where all edges point to that one.
We include GIDs from line graphs and cycle graphs (up to $100$K states in multiples of $3$); complete and complete acyclic graphs (up to $1$K states); and bipartite graphs (up to $1$K states).
These are important cases, for example, because the reverse-order line and cycle graphs are a potential worst case for Simple and BFGT.

Second, to exhibit more dynamic behavior, we generated random graphs:
sparse graphs with a fixed out-degree from each state, chosen from $1, 2, 3,$ or $10$ (up to $100$K states);
and dense graphs with a fixed probability of each edge, chosen from $.01, .02,$ or $.03$ (up to $10$K states).
Each case uses $10$ different random seeds.
As with the basic graphs, states are read in some order and marked closed.

To answer \textbf{Q3}, we wrote a backend to extract a GID at runtime from Z3's regex solver~\cite{stanfordsymbolic}.
While the backend of the solver is precisely a GID --- and so could be passed to our GID implementation dynamically ---
this setup includes many extraneous overheads,
including rewriting expressions and computing derivatives
when adding nodes to the graph.
While some of these overheads may be possible to eliminate,
and we are fairly confident that GIDs would be a bottleneck for sufficiently large input examples,
this makes it difficult to isolate the performance impact of the GID data structure itself, which is the sole focus of this paper.
We therefore instrumented the Z3 solver code to export the (incremental) sequence of graph updates that would be performed during a run of Z3 on existing regex benchmarks.
For each benchmark, this instrumented code produces a faithful representation of the sequence of graph updates that actually occur in a run of the SMT solver on this particular benchmark.
For each regex benchmark, we thus get a GID benchmark for the present paper.
The benchmarks focus on \emph{extended} regexes, rather than plain classical regexes as these are the ones for which dead state detection is relevant (see \Cref{sec:regex}).
We include GIDs for the RegExLib benchmarks~\cite{SMT12-regex} and the handcrafted Boolean benchmarks reported in~\cite{stanfordsymbolic}.
We add to these 11 additional examples designed to be difficult GID cases.
The collection of regex benchmarks we used (just described) is available \visiblehref{https://github.com/cdstanford/regex-smt-benchmarks}{on GitHub.}

From both the Q2 and Q3 benchmarks, we filter out any benchmark which takes under $10$ milliseconds for all of the algorithms to solve (including \Naive{}),
and we use a $60$ second timeout.
The evaluation was run on a 2020 MacBook Air (MacOS Monterey) with an Apple M1 processor and 8GB of memory.

\paragraph{Correctness.}
To ensure that all of our implementations our correct,
we invested time into unit testing and checked output correctness on all of our collected benchmarks, including several cases which exposed bugs in previous versions of one or more algorithms.
In total, all algorithms are vetted against
25 unit tests from handwritten edge cases that exposed prior bugs,
373 unit tests from benchmarks,
and 30 module-level unit tests for specific functions.

\paragraph{Results.}
\Cref{fig:eval-results} shows the results.
\Cref{alg:jump} shows significant improvements over the state-of-the-art,
solving more benchmarks in a smaller amount of time
across basic GIDs, random GIDs, and regex GIDs.
\Cref{alg:log} also shows state-of-the-art performance, similar
to BFGT on basic and regex GIDs and significantly better on random GIDs.
On the bottom right, since looking at average time is not meaningful for
benchmarks of widely varying size, we stratify the size of benchmarks into
buckets, and plot time-to-solve as a function of size.
Both $x$-axis and $y$-axis are on a log scale.
The plot shows that \Cref{alg:jump} exhibits up to two orders of magnitude speedup
over BFGT for larger GIDs --
we see speedups of $110$x to $530$x for GIDs in the top five size buckets
(GIDs of size nearest to $100$K, ${\sim}200$K, ${\sim}500$K, $1$M, and ${\sim}2$M).

\paragraph{New implementations of existing work.}
Our implementation contributes, to our knowledge, the first implementation of BFGT specifically for SCC maintenance.
In addition, it is one of the first implementations of Euler Tour Trees (see~\cite{BakaricET} for another), including the AVL tree backing for tours,
and likely the first implementation in Rust.

\section{Application to Extended Regular Expressions}
\label{sec:regex}

In this section, we explain how
precisely the GID state classification problem arises in the context
of derivative-based solvers~\cite{stanfordsymbolic,LTRTB15}.
We first define \emph{extended} regexes~\cite{gelade2008succinctness} (regexes extended with
intersection $\rand$ and complement $\rnot$) modulo a symbolic
alphabet $\A$ of \emph{predicates} that represent sets of characters.
We explain the main
idea behind \emph{symbolic derivatives}~\cite{stanfordsymbolic}
(a generalization of Brzozowski~\cite{Brzozowski64} and Antimirov derivatives~\cite{Ant95}, see also~\cite{CCM11,KeilT14})
that provides the foundation for incrementally creating a GID.
Then we show, through an example, how a solver can incrementally expand
derivatives to reduce the satisfiability problem to the GID state
classification problem (\Cref{def:problem}).

Define a \emph{regex} by the following grammar,
where $\varphi \in \A$ denotes a predicate:
\[
\REPRED ::= \;
  \varphi\;\mid\;
  \eps\;\mid\;
  \REPRED_1\cdot \REPRED_2 \;\mid\;
  \REPRED\st \;\mid\;
  \REPRED_1\ror \REPRED_2\;\mid\;
 \REPRED_1\rand \REPRED_2\;\mid\;
 \rnot \REPRED
 \]
Let $R^k$ represent the concatenation of $R$ $k$ times.
The \emph{symbolic derivative} of a regex $R$
(defined formally in \Cref{subsec:derivative-def})
describes the set of \emph{suffixes} of strings in $R$ after the first character is removed.

To apply \Cref{def:gid} to regexes:
\textbf{states} are regexes;
\textbf{edges} are transitions from a regex to its derivatives;
and
\textbf{terminal} states are the so-called \emph{nullable} regexes,
where a regex is nullable if it matches the empty
string. Nullability can be computed inductively over the structure of
regexes: for example, $\eps$ and $R\st$ are nullable,
and $R_1\rand R_2$ is nullable iff both $R_1$ and $R_2$ are nullable.
A \textbf{live} state here is thus a regex that reaches a nullable regex
via 0 or more edges. This implies that there exists a concrete string
matching it.
Conversely, \textbf{dead} states are always empty,
i.e. they match no strings, but can reach other dead states, creating
strongly connected components of closed states none of which are live.
For example, the \emph{false} predicate $\emp$ of $\A$ serves as the regex that
matches \emph{nothing} and is trivially a dead state.  Thus
$\rnot\emp$ is equivalent to $\TT\st$,
where $\TT$ is the \emph{true} predicate
and is trivially a live state.

\subsection{Formal definition}
\label{subsec:derivative-def}

A \emph{symbolic derivative} $\der{R}$ of a regex $R$
is a binary tree whose leaves are regexes and internal nodes are labeled by
predicates from $\A$.
A binary tree with root node labeled $\varphi$ and two immediate subtrees
$t_1$ and $t_2$ is written $\ite{\varphi}{t_1}{t_2}$;
it means that $\varphi$ is \emph{true} in $t_1$ and \emph{false} in $t_2$.
A branch with the \emph{accumulated branch condition} $\psi$ from the root of
$t$ to one of its leaves $R'$ defines a transition
$\tr{R}{\psi}{R'}$ -- \emph{provided $\psi$ is satisfiable in
$\A$} -- and the edge $(R,R')$ is added to the GID.
The number of
leaves of $\der{R}$ is the \emph{out-degree} $\OutDegree{R}$ of $R$
and is independent of the size of the concrete alphabet.
Thus, we label $R$ as \textbf{closed} when
$\OutDegree{R}$ edges have been added from $R$.

The formal definition of a symbolic
derivative is as follows where the operations $\Nror$, $\Nrand$, $\Nrnot$
and $\Ncdot$ are here for the purposes of this paper and {w.l.o.g.}
\emph{normalizing} variants of
$\ror$, $\rand$, $\rnot$ and $\cdot$,
by distributing the operations into if-then-elses and build
a single binary tree as a nested if-then-else as a result (see~\cite{stanfordsymbolic} for
details):
\[
\begin{array}{r@{\;}c@{\;}l@{\quad}r@{\;}c@{\;}l@{\quad}r@{\;}c@{\;}l}
  \der{\eps}&=&\der{\emp} = \emp &
  \der{\varphi} &=& \ite{\varphi}{\eps}{\emp}&
  \der{R*} &=& \der{R}\Ncdot R{*}
  \\
  \der{R\ror R'} &=& {\der{R}}\Nror{\der{R'}} &
  \der{R\rand R'} &=& {\der{R}}\Nrand{\der{R'}} &
  \der{\rnot R} &=& \Nrnot{\der{R}}
  \\
  \der{R\cdot R'} &=&
  \multicolumn{7}{l}{\IfThenElse{\Nullable{R}}{(\der{R} \Ncdot R')\Nror{\der{R'}}}{\der{R} \Ncdot R'}}
\end{array}
\]
If $c$ is a term of type \emph{character} then $\der[c]{R}$ is
obtained from $\der{R}$ by instantiating all internal nodes
(conditions of the if-then-elses) $\varphi$ by tests $c\in\varphi$. This means that in the
context of SMT, $\der[c]{R}$ is a term of type \emph{regex}, while
$\der{R}$ represents the lambda-term $\lambda c.\der[c]{R}$ that is
constructed \emph{independently} of $c$.

\subsection{Reduction from Incremental Regex Emptiness to GIDs}

For simplicity, suppose we want to
determine the satisfiability of a single regex constraint $s \in R$,
where $s$ is a string variable and $R$ is a concrete regex.
(This is not overly restrictive -- any number of simultaneous regex constraints for a string
$s$ can be combined into single regex constraint by using the Boolean
operations of regexes.)
For example, let
$L = \rnot(\TT\st \alpha \TT^{100})$ and $R= L \rand (\TT \alpha)$,
where $\alpha$ is the ``is digit'' predicate that is true of characters that are digits
(often denoted \texttt{\char`\\d}).
The solver
manipulates regex membership constraints on strings by unfolding
them~\cite{stanfordsymbolic}.  The constraint $s \in R$, that
essentially tests nonemptiness of $R$ with $s$ as a witness, becomes
\[
(s = \epsilon \land \Nullable{R}) \lor (s \neq \epsilon
\land \Tail{s}\in\der[{\ith{s}{0}}]{R})
\]
where, $s\neq \epsilon$ since $R$ is not nullable, $\Tail[i]{s}$ is the
suffix of $s$ from index $i$, and
\[
\der{R} = \der{L}\Nrand\der{\TT\alpha} =
\ite{\alpha}{L\rand\rnot(\TT^{100})}{L} \Nrand \alpha
=
\ite{\alpha}{{L\rand\rnot(\TT^{100})\rand\alpha}}{{L\rand\alpha}}
\]
Let $R_1=L\rand\rnot(\TT^{100})\rand\alpha$ and $R_2=L\rand\alpha$.
So $R$ has two outgoing transitions $\tr{R}{\alpha}{R_1}$ and
$\tr{R}{\lnot{\alpha}}{R_2}$ that contribute the edges $(R,R_1)$ and
$(R,R_2)$ into the GID. Note that these edges depend only on $R$ and
not on $\ith{s}{0}$.

We continue the search incrementally by checking the two branches of the if-then-else constraint,
where $R_1$ and $R_2$ are again not nullable (so $\Tail{s} \neq \epsilon$):
\[
\begin{array}{l}
  \ith{s}{0}\in\alpha \land \Tail[2]{s}\in\der[{\ith{s}{1}}]{R_1} \quad\lor\quad
  \ith{s}{0}\in\lnot\alpha \land \Tail[2]{s}\in\der[{\ith{s}{1}}]{R_2} \\[.4em]
  \der{R_1}= \ite{\alpha}{L\rand\rnot(\TT^{100})\rand\rnot(\TT^{99})}{L\rand\rnot(\TT^{99})}
  \Nrand\ite{\alpha}{\eps}{\emp} = \ite{\alpha}{\eps}{\emp} \\
  \der{R_2}= \ite{\alpha}{L\rand\rnot(\TT^{100})}{L}
  \Nrand\ite{\alpha}{\eps}{\emp} = \ite{\alpha}{\eps}{\emp} \\
\end{array}
\]
It follows that $\tr{R_1}{\alpha}{\eps}$ and $\tr{R_2}{\alpha}{\eps}$,
so the edges $(R_1,\eps)$ and $(R_2,\eps)$ are added to the GID where
$\epsilon$ is a trivial terminal state. In fact, after $R_1$ the
search already terminates because we then have the path
$(R,R_1)(R_1,\epsilon)$ that implies that $R$ is live. The associated
constraints $\ith{s}{0}\in\alpha$ and $\ith{s}{1}\in\alpha$ and the
final constraint that $\Tail[2]{s}=\epsilon$ can be used to extract a concrete
witness, e.g., $s=\texttt{"42"}$.

\emph{Soundness} of the algorithm follows from that if $R$ is
nonempty ($s\in R$ is \emph{satisfiable}), then we eventually arrive at a nullable (terminal) regex, as in the
example run above.  To achieve \emph{completeness} -- and to eliminate
dead states as early as possible -- we incrementally construct a GID
corresponding to the set of regexes seen so far (as above).  After all
the feasible transitions from $R$ to its derivatives in $\der{R}$ are added to the
GID as edges (WLOG in one batch), the state $R$ becomes closed.
\emph{Crucially, due to the \textbf{symbolic form} of
$\der{R}$, no derivative is missing.}
Therefore $R$ is known to be
empty precisely as soon as $R$ is detected as a dead state in the GID.
An additional benefit is that the algorithm is independent
of the size of the universe of $\A$, that may be very large (e.g. the
Unicode character set), or even infinite.
We get the following theorem that uses finiteness of the closure
of symbolic derivatives~\cite[Theorem~7.1]{stanfordsymbolic}:
\begin{theorem}
For any regex $R$,
(1) If $R$ is nonempty, then the decision procedure eventually marks $R$ live.
(2) If $R$ is empty, then the decision procedure marks $R$ dead
at the earliest stage that it is know to be dead, and terminates.
\end{theorem}

\section{Related Work}
\label{sec:rw}

\paragraph{Online graph algorithms.}
Online graph algorithms are typically divided into problems over
\emph{incremental} graphs (where edges are added), \emph{decremental}
graphs (where edges are deleted), and \emph{dynamic} graphs (where
edges are both added and deleted), with core data structures discussed
in~\cite{Mehlhorn84,eppstein1999dynamic}.
Important problems include \emph{transitive closure}, \emph{cycle detection},
\emph{topological ordering}, and \emph{strongly connected component (SCC) maintenance}.

For incremental topological ordering, \cite{MNR96} is an early work,
and \cite{HKMST12} presents two different algorithms, one
for \emph{sparse graphs} and one for \emph{dense graphs} -- the
algorithms are also extended to work with SCCs. The sparse algorithm
was subsequently simplified in~\cite{BFGT15} and is the basis of our
implementation named BFGT in \Cref{sec:eval}.
A unified approach
of several algorithms based on~\cite{BFGT15} is presented in~\cite{Cohen13} that uses a notion of \emph{weak topological order}
and a labeling technique that estimates transitive closure size.
Further extensions of~\cite{BFGT15} are studied
in~\cite{BerChe18,Bhattacharya2020} based on randomization.

For dynamic directed graphs,
a topological sorting algorithm that is experimentally preferable
for sparse graphs is discussed in~\cite{PK06}, and a related
article~\cite{PK04} discusses strongly connected components maintenance.
Transitive closure for dynamic graphs is studied
in~\cite{RZ08}, improving upon some algorithms presented earlier
in~\cite{HK95}.
One major application for these algorithms is in pointer analysis~\cite{Pearce05}.

For \emph{undirected} forests, fully dynamic reachability is solvable in amortized logarithmic time per edge via multiple possible approaches~\cite{sleator1983data,frederickson1997data,henzinger1999randomized,alstrup2005maintaining,tarjan2010dynamic}; our implementation uses Euler Tour Trees~\cite{henzinger1999randomized}.

\paragraph{Data structures for SMT.}
\emph{UnionFind}~\cite{Tar75} is a foundational data structure used in SMT.
\emph{E-graphs}~\cite{de2007efficient,willsey2021egg}
are used to ensure \emph{functional extensionality}, where two
expressions are equivalent if their subexpressions are equivalent~\cite{downey1980variations,nelson1980fast}.
In both UnionFind and E-graphs, the maintained relation is an \emph{equivalence} relation.
In contrast, maintaining live and dead states involves tracking
reachability rather than equivalence.
To the best of our knowledge, the
specific formulation of incremental reachability we consider here is new.

\paragraph{Dead state elimination in automata.}
A DFA or NFA may be viewed as a GID, so state classification in GIDs
solves dead state elimination in DFAs and NFAs,
while additionally working in an incremental fashion.
Dead state elimination is also known as \emph{trimming}~\cite{HopUll79}
and plays an important role in automata \emph{minimization}~\cite{Hop71,BBCF11,MaCl13}.
The literature on minimization is vast,
and goes back to the 1950s~\cite{Huff54,Moore56,Brz63,Hopcroft69,KW70,PT87,Blum96};
see \cite{Watson93} for a taxonomy,
\cite{Almeida07} for an experimental comparison,
and \cite{DV14} for the symbolic case.
Watson et. al.~\cite{Watson03}
propose an \emph{incremental} minimization
algorithm, in the sense that it can be halted at any point to produce a
partially minimized, equivalent DFA;
unlike in our setting, the DFA's states and transitions are fixed and read in a predetermined order.

\section{Conclusion}
\label{sec:conclusion}

We have studied \visiblehref{https://github.com/cdstanford/gid}{guided incremental digraphs:} incremental transition systems in which states are
\emph{closed} when they will receive no further outgoing edges.
\Cref{alg:log,alg:jump} solve the incremental live and dead state detection problem in GIDs.
The former runs in amortized logarithmic time, and both algorithms achieve orders of magnitude speedup over the state-of-the-art based on maintaining strong connected components.
We anticipate a wide range of future applications: to build better regex solvers; to build lazy decision procedures for logical problems such as LTL model checking; and to build new incremental program analyses for properties such as termination and liveness.
We hope that this work inspires others to apply online graph algorithms to problems in verification -- especially, to obtain lazy and incremental algorithms in cases where brute-force search is prohibitive.

\section*{Acknowledgments}

We thank the anonymous reviewers of CAV 2021, TACAS 2022, and CAV 2023 for feedback leading to substantial improvements to both our paper and our results.
Special thanks to
Nikolaj Bj{\o}rner, for his collaboration and involvement with Z3,
and Yu Chen, for helpful discussions in which he proposed the idea for
the first-cut algorithm.

\bibliographystyle{splncs04}
\bibliography{ref}

\end{document}